\PassOptionsToPackage{dvipsnames}{xcolor}
\documentclass[sigplan,screen,nonacm]{acmart}

\usepackage{mathpartir}
\usepackage{pgf}
\usepackage{stmaryrd}

\usepackage{cleveref}
\crefname{lstlisting}{listing}{listings}
\Crefname{lstlisting}{Listing}{Listings}
\crefname{line}{line}{lines}
\Crefname{line}{Line}{Lines}

\newcommand{\commentNote}[1]{}

\usepackage{color}
\definecolor{keywordcolor}{rgb}{0.7, 0.1, 0.1}   %
\definecolor{tacticcolor}{rgb}{0.0, 0.1, 0.6}    %
\definecolor{commentcolor}{rgb}{0.4, 0.4, 0.4}   %
\definecolor{symbolcolor}{rgb}{0.0, 0.1, 0.6}    %
\definecolor{sortcolor}{rgb}{0.1, 0.5, 0.1}      %
\definecolor{attributecolor}{rgb}{0.7, 0.1, 0.1} %

\usepackage{listings}

\DeclareUnicodeCharacter{2080}{\ensuremath{_0}}
\DeclareUnicodeCharacter{2081}{\ensuremath{_1}}
\DeclareUnicodeCharacter{2082}{\ensuremath{_2}}
\DeclareUnicodeCharacter{2083}{\ensuremath{_3}}
\DeclareUnicodeCharacter{2084}{\ensuremath{_4}}
\DeclareUnicodeCharacter{2085}{\ensuremath{_5}}
\DeclareUnicodeCharacter{2086}{\ensuremath{_6}}
\DeclareUnicodeCharacter{2087}{\ensuremath{_7}}
\DeclareUnicodeCharacter{2088}{\ensuremath{_8}}
\DeclareUnicodeCharacter{2089}{\ensuremath{_9}}
\DeclareUnicodeCharacter{1D62}{\ensuremath{_i}}
\DeclareUnicodeCharacter{2C7C}{\ensuremath{_j}}
\DeclareUnicodeCharacter{2090}{\ensuremath{_a}}
\DeclareUnicodeCharacter{2099}{\ensuremath{_n}}
\DeclareUnicodeCharacter{2098}{\ensuremath{_m}}
\DeclareUnicodeCharacter{209A}{\ensuremath{_p}}
\DeclareUnicodeCharacter{00B9}{\ensuremath{^1}}

\setcopyright{acmlicensed}
\copyrightyear{2018}
\acmYear{2018}
\acmDOI{XXXXXXX.XXXXXXX}

\acmJournal{JACM}
\acmVolume{37}
\acmNumber{4}
\acmArticle{111}
\acmMonth{8}

\author{Timon Böhler}
\orcid{0009-0002-9964-7367}
\affiliation{%
  \institution{TU Darmstadt}
  \city{Darmstadt}
  \country{Germany}
}
\email{timon.boehler@tu-darmstadt.de}

\author{Simon Daniel}
\orcid{0009-0003-1636-7253}
\affiliation{%
  \institution{TU Darmstadt}
  \city{Darmstadt}
  \country{Germany}
}
\email{simon.daniel@tu-darmstadt.de}

\author{David Richter}
\orcid{0000-0002-8672-0265}
\affiliation{%
  \institution{TU Darmstadt}
  \city{Darmstadt}
  \country{Germany}
}
\email{david.richter@tu-darmstadt.de}

\author{Pascal Weisenburger}
\orcid{0000-0003-1288-1485}
\affiliation{%
  \institution{University of St. Gallen}
  \city{St. Gallen}
  \country{Switzerland}
}
\email{pascal.weisenburger@unisg.ch}

\author{Mira Mezini}
\orcid{0000-0001-6563-7537}
\affiliation{%
  \institution{TU Darmstadt}
  \city{Darmstadt}
  \country{Germany}
}
\affiliation{%
  \institution{hessian.AI}
  \city{Darmstadt}
  \country{Germany}
}
\affiliation{%
  \institution{National Research Center for Applied Cybersecurity ATHENE}
  \city{Darmstadt}
  \country{Germany}
}
\email{mezini@informatik.tu-darmstadt.de}

\begin{document}

\title{Mechanizing Choreographic Programs and Hoare Logic with State Transformers}

\ccsdesc[500]{Software and its engineering~Distributed programming languages}
\ccsdesc[300]{Theory of computation~Hoare logic}
\ccsdesc[500]{Theory of computation~Program verification}
\keywords{choreographic programming, Lean, interactive theorem proving, mechanization}

\newcommand{\chorColor}[1]{\mathsf{\textcolor{RoyalBlue}{#1}}}
\newcommand{\procColor}[1]{\textbf{\textcolor{RedOrange}{#1}}}
\newcommand{\sigs}{\ensuremath{\mathsf{Method}}}
\newcommand{\interp}{\ensuremath{\mathsf{I}}}
\newcommand{\procedure}{\ensuremath{\mathsf{Func}}}
\newcommand{\pden}{\ensuremath{\mathcal{D}}}

\newcommand{\chor}[0]{\mathsf{Chor}}
\newcommand{\eff}[2]{\ensuremath{\chorColor{run}_{#1}(#2)}}
\newcommand{\com}[4]{\ensuremath{\chorColor{com}_{#1,#2}(#3, #4)}}
\newcommand{\switch}[2]{\ensuremath{\chorColor{bcast}_{#1}(#2)}}
\newcommand{\done}{\ensuremath{\chorColor{done}}}
\newcommand{\call}[1]{\ensuremath{\chorColor{call}(#1)}}
\newcommand{\seq}{\ensuremath{~\chorColor{\fatsemi}~}}

\newcommand{\stat}[2]{\langle #1 ~\vert~ #2 \rangle}
\newcommand{\effl}[1]{\ensuremath{\mathsf{run}_{#1}}}
\newcommand{\sendl}[2]{\ensuremath{\mathsf{com}_{#1,#2}}}
\newcommand{\broadl}[1]{\ensuremath{\mathsf{bcast}_{#1}}}
\newcommand{\calll}{\ensuremath{\mathsf{call}}}
\newcommand{\labelC}[1]{\ensuremath{\chorColor{\to}_{#1}}}
\newcommand{\labelN}[1]{\ensuremath{\textcolor{RedOrange}{\leadsto}_{#1}}}
\newcommand{\stepC}[5]{\stat{#1}{#2}~\labelC{#3}~\stat{#4}{#5}}
\newcommand{\stepN}[5]{\stat{#1}{#2}~\labelN{#3}~\stat{#4}{#5}}
\newcommand{\stepM}[5]{\stat{#1}{#2}~\chorColor{\twoheadrightarrow}^{#3}~\stat{#4}{#5}}
\newcommand{\stepCS}[4]{\stat{#1}{#2}~\chorColor{\to}^\star~\stat{#3}{#4}}
\newcommand{\stepNS}[4]{\stat{#1}{#2}~\textcolor{RedOrange}{\leadsto}^\star~\stat{#3}{#4}}

\newcommand{\effp}[1]{\ensuremath{\procColor{run}(#1)}}
\newcommand{\send}[2]{\ensuremath{\procColor{send}_{#1}(#2)}}
\newcommand{\recv}[2]{\ensuremath{\procColor{recv}_{#1}(#2)}}
\newcommand{\broad}[1]{\ensuremath{\procColor{shout}(#1)}}
\newcommand{\switchp}[1]{\ensuremath{\procColor{listen}_{#1}}}
\newcommand{\donep}[0]{\ensuremath{\procColor{\textup{done}}}}
\newcommand{\proc}[1]{\ensuremath{\mathsf{Proc}_{#1}}}
\newcommand{\callp}[1]{\ensuremath{\procColor{call}(#1)}}
\newcommand{\seqp}{\ensuremath{~\textcolor{RedOrange}{\fatsemi}~}}

\newcommand{\proj}[1]{\ensuremath{\llbracket #1 \rrbracket}}

\begin{abstract}

Choreographic programming is a programming model for developing distributed applications where an entire communication protocol is written as a single program, which a compiler then projects to one process per participant. Choreographic programming abstracts over low-level network communication primitives such as sockets, and provides a high degree of safety guarantees with deadlock freedom ensured by construction.

Mechanizing choreographies necessarily deals with both operations specific to distributed programming and standard (local) operations that also occur in non-distributed programs, as well as the typical issues of binding and substitution. We aim to sidestep the latter issues, thereby obtaining a more concise mechanization that focuses on the essential distributed aspects of choreographies.

To this end, we use a method recently proposed by Thiemann to elegantly model deadlock-free processes in a dependently typed language: Using state transformers to represent the computations performed by each process.

We bring the state transformer model to choreographies, allowing us to reduce the usual mechanization effort around binding and substitution, and to abstract over the details of the ``local'' aspects of the language.

We mechanize in Lean a choreographic language that supports point-to-point communication, broadcasting, recursive procedures, and \emph{local stateful methods}, allowing each participant to be assigned a different set of methods. We prove soundness and completeness of endpoint projection, establish deadlock freedom for the projected processes, prove confluence, and verify a Hoare logic for choreographies.

\end{abstract}

\maketitle

\section{Introduction}

Distributed programs are notoriously hard to write. Two major reasons are the complexity of managing communication between multiple programs and the
risk of deadlocks. Choreographic programming~\cite{Montesi23} promises to alleviate these issues by allowing distributed programs to be described using a
\emph{global} view where the behavior of all participants is contained in a single program.
To execute choreographies, they are \emph{projected} to a set of communicating processes (or \emph{endpoints}).
The main correctness properties one wants to ensure are that (i) endpoint projection preserves the semantics of the choreography and (ii) the projected network is free from deadlocks.

While many choreographic languages with wildly differing feature sets have been proposed, \emph{mechanizations} of choreographies are rarer,
and are often quite complex, relying on intricate proof developments~\cite{CruzFilipeMP23}.
For instance, a recent mechanization of multiparty session types~\cite{Castro-PerezFJ26} notes that for the most difficult theorem, ``the majority of the proof is about dealing with binders, renaming, etc.''.
In this work, we demonstrate an encoding of choreographies using \emph{state transformers}, which have previously been applied by
Thiemann~\cite{Thiemann23} to the modeling of processes in the context of session types.
The idea is that each process is assigned a state (of an arbitrary type). Operations like sending
and receiving are then modelled by meta-language functions which read the state to produce a message
or which read a message and return a new state.
This encoding has the advantage that it does not require a context of variables for each participant, avoids issues relating to naming
and binding, and does not require explicitly defining a notion of locally evaluated expressions, as the embedding automatically supports arbitrary meta-language
expressions. We therefore get a neat separation between the essential, distributed features of choreographies and the local aspects that do not
influence the proof of correctness. Additionally, embedded programs using state
transformers tend to be more readable than approaches using encodings
like de Bruijn indices or locally nameless.

Using this encoding, we provide a mechanized proof of the soundness and completeness of endpoint projection and of deadlock freedom of projected networks.
The choreographies we describe can (i) send messages between two participants, (ii) switch between branches of the choreography by communicating to each participant a message stating which branch was chosen, (iii) call (potentially recursive) procedures, and (iv) perform role-dependent local operations.
The last feature is formalized by mapping a set of state-transformer functions to each role, allowing the set of methods to differ
between different participants (for example, a participant might have a method for modifying a local database).
To our knowledge, this feature has not previously been formalized, and the ease of mechanizing it underscores the benefits of the state-transformer approach.

We further mechanize a proof of confluence for the choreography semantics and a partial correctness Hoare logic, enabling formal reasoning about choreographies.

Our contributions are:
\begin{itemize}
\item A formalization of choreographies via state transformers, which avoids tedious handling of naming and variable binding.
\item A mechanized proof of soundness and completeness of endpoint projection, of deadlock freedom for projected networks, of partial correctness for an interpreter, of confluence for choreographies, and of a partial correctness Hoare logic.
\item The introduction and mechanization of role-dependent methods.
\end{itemize}

The paper is structured as follows: In Section~\ref{sec:bg}, we introduce choreographies, go through
the set of features our choreographic calculus supports, and explain the embedding using state transformers. Section~\ref{s:mech} describes the formalization of choreographies and processes, and of their small-step operational semantics and of projection. More importantly, it contains the soundness, completeness and deadlock freedom proofs, and the proof of partial correctness for an interpreter with respect to the operational semantics.
We discuss related work in Section~\ref{sec:rel} before concluding in Section~\ref{sec:conc}.

\paragraph{Colors}

Following Patrignani's~\cite{Patrignani2020} recommendation to visually distinguish the syntax of different languages in the same paper, we distinguish choreographies and processes through different colors and different fonts: Choreographic syntax is in a blue sans-serif font (e.g.: \done), while process syntax is in a red bold serif font (e.g.: \donep).

An artifact containing the mechanization is available~\cite{boehler26}.

\section{Representing Processes and Choreographies}
\label{sec:bg}

\subsection{State Transformers}

When participants in a choreography receive messages or call
local methods, their state changes. Standard approaches model the receiving of messages by extending the participant's environment, i.e. by binding a new variable.
In the state transformer approach, in contrast, the participant's state is a value of an arbitrary type (which can be different between each participant).

Communication in our language is formalized as follows.

\[
  \inferrule*{p,q : \mathsf{Role} \\ p \neq q \\ c : \chor{}\\\\ 
    v : \mathsf{LState}_p \to m \\
    u : \mathsf{LState}_q \times m \to \mathsf{LState}_q}
    {\com{p}{q}{v}{u}; c : \chor{}}
\]

A communication from role \lstinline{p} to role \lstinline{q}
is defined by two functions: \lstinline{v} describes how
\lstinline{p} computes the message that it sends from its own state, while \lstinline{u} describes how \lstinline{q} updates
its own state upon receiving the message.
Note that the message type \lstinline{m} is a parameter of the
\lstinline{com} constructor, meaning that each occurrence of
\lstinline{com} in a program can have a different message type.

We describe a simple protocol between two participants Alice and Bob,
in which Alice tries to guess Bob's secret number, upon which he tells her
whether she guessed correctly. Each participant's state is a Boolean. 
First, Alice sends a message containing her $\mathrm{guess}$; upon receiving the message,
Bob compares it to his $\mathrm{secret}$ and tells Alice about the result.

\[\begin{array}{l}
\com{\mathsf{alice}}{\mathsf{bob}}{(\lambda\_.~\mathrm{guess})}{(\lambda s~g.~g = \mathrm{secret})}; \\
\com{\mathsf{bob}}{\mathsf{alice}}{(\lambda b.~b)}{(\lambda s~b.~b)}; \done{}
\end{array}\]

We show the Lean code below. Note that calls to \lstinline|.com| require a
proof that the participants are distinct. In the example below, this proof is derived automatically using implicit arguments that call Lean's \lstinline|decide|
tactic which can discharge decidable propositions (in this case, that the two constant values \lstinline|.alice| and \lstinline|.bob| are different). Note that \lstinline!<|! is low-precedence function application (like \lstinline!$! in Haskell).

\begin{lstlisting}
.com .alice .bob (λ _ => guess) (λ _ m => m == secret) <|
.com .bob .alice (λ b => b) (λ _ b => b) <|
.done
\end{lstlisting}

\paragraph{Simulating variables.}

As mentioned, we use state transformers instead of variable contexts.
We now show how variable context can be simulated in our approach.
In addition, it is easy to see that no process can read another process' data, as the functions all operate on the local state only. When using normal contexts, one has to explicitly keep track which
variables are accessible from which process.

Variables can be modeled by instantiating each process' local state with its variable context.
Consider the following example where the process $p$'s state has type $\{a: \mathbb{N}, b: \mathbb{N}\}$ and the process $q$'s state has type $\{x: \mathbb{N}, y: \mathbb{N}\}$.
The choreography first has $p$ send $a+1$ to $q$, which $q$ stores in $x$.
Then $q$ sends $x+y$ to $p$, which $p$ stores in $b$.
We write the program in our representation; below, we show
pseudocode (inspired by HasChor~\cite{ShenKK23}) for comparison.
\[
\begin{array}{l}
\com{p}{q}{(\lambda s.~s.a+1)}{(\lambda s~x'.~s[x \mapsto x'])}; \\
\com{q}{p}{(\lambda s.~s.x+s.y)}{(\lambda s~b'.~s[b \mapsto b'])}; \done{}
\end{array}
\]
\begin{lstlisting}
         let x ← (p, a + 1) ~> q
         let b ← (q, x + y) ~> p
         return ()
\end{lstlisting}

While most choreographic languages use a functional style where messages introduce new variable bindings,
our representation using state transformers is also very similar to the language used by Cruz{-}Filipe and Montesi~\cite{Cruz-FilipeM16b} (but not mechanized) to write \emph{imperative} algorithms as
choreographies.

\subsection{Broadcasting}

Our choreographic language supports branching. For branching to be projectable to a network of processes,
it is necessary that every involved process knows which branch has been taken.
We follow HasChor's approach in adding a $\switch{p}{v}; c$ construct, in which the process $p$ transmits the
same message to all participants (computed by $v$), which is then used to choose which branch to take.
Here, $c$ is a continuation which takes the message's value and returns a choreography (i.e., syntax).
The idea of representing the continuation of an operation as a meta-language from values to syntax
is widely used in embedding effectful languages~\cite{Chlipala21}, like in Freer Monads~\cite{Kiselyov15} or Interaction trees~\cite{XiaZHHMPZ20}. Broadcasting can hence be seen as a globally executed effect,
whose result determines how the choreography continues.

We demonstrate broadcasting below for an example where the message type is Boolean, meaning we choose between
two branches. Again, Alice tries to guess Bob's secret, but this time the result decides whether
Bob sends Alice a message in return.

\[
\begin{array}{l}
\com{\mathsf{alice}}{\mathsf{bob}}{(\lambda\_.~\mathrm{guess})}{(\lambda\_.~\lambda m.~m = \mathrm{secret})}; \\
\switch{\mathsf{bob}}{(\lambda s.~s)}; \\
\lambda m.~\begin{cases}
  \com{\mathsf{bob}}{\mathsf{alice}}{(\lambda\_.~\mathrm{"hello"})}{(\lambda s\_.~s)}; \done{} & \text{if } m \\
  \done{} & \text{if } \neg m
\end{cases}
\end{array}
\]

In Lean, the example looks as follows:

\begin{lstlisting}
.com .alice .bob (λ _ => guess) (λ _ m => m == secret) <|
.bcast .bob (λ s => s) λ
  | true => .com .bob .alice (λ _ => "hello") (λ s _ => s) .done
  | false => .done
\end{lstlisting}

\subsection{Methods}

In a distributed network, different participants may have different capacities, for example, one participant may have access to a database that others do not have access to. We model this using role-dependent methods. More precisely, our approach allows developers to define a type family \lstinline{Method} mapping each role to a type representing the role's methods. They then define a function \lstinline{I} which maps each method to its interpretation, which is a function that takes the role's state and returns a new state.

In the example below, the participant Bob has access to a method \lstinline{randNat} (implemented using the Lean function of the same name), allowing him to generate random numbers. His state consists of the most recently generated number and the state used by the number generator. The choreography consists of Bob generating a random number and passing it to Alice, who adds it to her current
state.

\[\begin{array}{l}
\eff{\mathsf{bob}}{\mathsf{randNat}(1, 100)}; \\
\com{\mathsf{bob}}{\mathsf{alice}}{\lambda(n. ~n)}{(\lambda s.~\lambda n.~s + n)}; \done{}
\end{array}\]

To define the choreography in Lean, we first define Bob's type of methods \lstinline{MethodBob}.
We then define the choreography's signature (more on that in Section~\ref{s:mech}),
which means specifying the state type for each process, as well as the available methods
and their interpretation.

\begin{lstlisting}
inductive MethodBob where | randNat (lo hi : Nat)

instance RandEx : Sig where -- unused parameters omitted
  LState | .alice => Nat   | .bob => Nat × StdGen
  Method | .alice => Empty | .bob => MethodBob
  I | .bob, .randNat lo hi, (_, gen) => randNat gen lo hi

def randex : @Chor RandEx :=
  .run .bob (.randNat 1 100) <|
  .com .bob .alice (λ (n, _) => n) (λ (s n : Nat) => s + n)
  .done
\end{lstlisting}

\paragraph{Procedures}

Similar to methods are procedures, which however evaluate to choreographies, can be recursive,
and are called globally. The syntax is $\call{f}; c$, where $f$ is a procedure name.

\subsection{Interpretation}
\label{s:int}

We also give an interpreter for our state-transformer-based choreographies.
The interpreter below takes a set of procedures, a global state,
and a choreography. It evaluates the choreography (deterministically) and outputs the final global state.
Here, $\pden$ is a function mapping each procedure to its function body, $g$ contains each processes state,
$\seq$ sequentially composes two choreographies, $g[p \mapsto x]$ sets the state of process $p$
to $x$, and $\mathsf{I}_p(e, g_p)$ runs method $e$ on the local state of process $p$.

\[
\begin{array}{rcl}
\mathcal{I}(\pden, g, \eff{p}{e}; c) &\triangleq& \mathcal{I}(\pden, g[p \mapsto \mathsf{I}_p(e, g_p)], c) \\
\mathcal{I}(\pden, g, \com{p}{q}{v}{u}; c) &\triangleq& \mathcal{I}(\pden, g[q \mapsto u(g_q, v(g_p))], c) \\
\mathcal{I}(\pden, g, \switch{p}{v}; c) &\triangleq& \mathcal{I}(\pden, g, c(v(g_p))) \\
\mathcal{I}(\pden, g, \call{f}; c) &\triangleq& \mathcal{I}(\pden, g, \pden(f) \seq c) \\
\mathcal{I}(\pden, g, \done{}) &\triangleq& g
\end{array}
\]

We have implemented a version of this interpreter, handling nontermination using a fuel argument and the \lstinline{Option} monad (see Section~\ref{sec:interp}), for which we
have proven partial correctness, i.e., if it returns some final state, that state is reachable in the operational semantics. All our example programs also work with the verified fuel-based interpreter.

\section{Mechanization}
\label{s:mech}

In this section, we mechanize choreographies and processes.
Our main result is deadlock freedom, which is a central guarantee provided by choreographic programming.

\subsection{Signature}

Our development is parametrized over
(i) a %
  type $\mathsf{Role}$ of process identifiers, which is assumed to have decidable equality,
(ii) a type $\mathsf{LState}$ of local states for each process identifier,
(iii) a type $\sigs{}$ of methods for each process identifier that are available there locally,
(iv) a type $\procedure{}$ of procedure identifiers, and
(v) a function $\interp{}$ interpreting each method as a state transformer.

\[\begin{array}{l}
\mathsf{Role} : \mathsf{Type} \qquad
\mathsf{LState} : \mathsf{Role} \to \mathsf{Type} \\
\sigs{} : \mathsf{Role} \to \mathsf{Type} \qquad
\procedure{} : \mathsf{Type} \\
\interp{} :  (p : \mathsf{Role}) \to \sigs{}_p \to \mathsf{LState}_p \to \mathsf{LState}_p
\end{array}\]

We additionally require $\mathsf{Role}$ to be non-empty (\lstinline!Inhabited!) and to have decidable equality, as made explicit by the Lean mechanization below, where the parameters are captured in a type class, to enable them being passed implicitly, reducing notational overhead.

\begin{lstlisting}
class Sig where
  Role : Type
  (inhRole : Inhabited Role) (decRole : DecidableEq Role)
  LState : Role → Type
  Method : Role → Type
  Func : Type
  I : (p : Role) → Method p → LState p → LState p
\end{lstlisting}

\subsection{Syntax}
\paragraph{Choreographies.}

\begin{figure}
\begin{mathpar}
  \inferrule*{p : \mathsf{Role} \\ e : \sigs{}(p) \\ c : \chor{}}
    {\eff{p}{e}; c : \chor{}}

  \inferrule*{p,q : \mathsf{Role} \\ p \neq q \\ c : \chor{}\\ 
    v : \mathsf{LState}_p \to m \\
    u : \mathsf{LState}_q \times m \to \mathsf{LState}_q}
    {\com{p}{q}{v}{u}; c : \chor{}}

  \inferrule*{p : \mathsf{Role} \\\\ v : \mathsf{LState}_p \to m \\\\ c : m \to \chor{}}
    {\switch{p}{v}; c : \chor{}}

  \inferrule*{c : \chor{} \\ f : \mathsf{Func}}
    {\call{f}; c : \chor{}}

  \inferrule*{ }
    { \done{} : \chor{}}
\end{mathpar}
\caption{Intrinsically typed syntax of choreographies.}
\label{fig:choreographies}
\end{figure}

We model choreographies as sequential programs featuring message passing, branching,
local methods, and recursive procedures.
The syntax and typing rules for choreographies are given in \cref{fig:choreographies}.
We formalize choreographies using an inductive type $\chor$, which is
parametrized over a signature $\sigs$, giving each process a set of methods.
$\chor$ is constructed by one of the following rules:

\begin{itemize}
\item $\eff{p}{e};  c$ -- Run a method $e$ on process $p$ from the process-specific method set ${\sigs{}_p}$, then continue with $c$.
\item $\com{p}{q}{v}{u}; c$ -- Send a message (computed from the sender's state by $v : \mathsf{LState}_p \to m$) from process ${p}$ to $q$ (with proof that ${p \neq q}$).
After communicating, process ${q}$ updates its local state using ${u : \mathsf{LState}_q \to m \to \mathsf{LState}_q}$, and the choreography continues as $c$.

\item $\switch{p}{v}; c$ -- Broadcast from process ${p}$, sending a message computed by ${v : \mathsf{LState}_p \to m}$, with the continuation depending on the message, ${c: m \to \chor}$.
The broadcast involves all participants, following the approach taken by HasChor~\cite{ShenKK23}.

\item $\call{f}; c$ -- Call a procedure $f$ and continue with $c$.

\item $\done{}$ -- The terminal choreography, indicating no further behavior.
\end{itemize}

In Lean, the syntax is defined as an inductive type, parametrized over a signature \lstinline{sig}. Notably, we do not need a notion of variable context, instead using meta-language functions in \lstinline|com| and \lstinline|bcast|.

\begin{lstlisting}
inductive Chor [sig : Sig] : Type 1
| run (p : sig.Role) (eff : sig.Methods p) (c : Chor)
| com {m : Type} (p q : sig.Role) (h : p ≠ q)
  (v : sig.LState p → m)
  (u : sig.LState q → m → sig.LState q) (c : Chor)
| bcast {m : Type} (p : sig.Role)
  (v : sig.LState p → m) (c : m → Chor)
| call (x : sig.Func) (c : Chor)
| done
\end{lstlisting}

\paragraph{Processes.}

\begin{figure} 
\begin{mathpar}
  \inferrule*{ e : \sigs{}(p) \\ b : \proc{p} }
   { \effp{e}; b : \proc{p} }

  \inferrule*{ q : \mathsf{Role} \quad b : \proc{p} \\\\ v : \mathsf{LState}_p \to m }
   { \send{q}{v}; b : \proc{p} }

  \inferrule*{ q : \mathsf{Role} \quad b : \proc{p} \\\\ u : \mathsf{LState}_p \times m \to \mathsf{LState}_p }
   { \recv{q}{u}; b : \proc{p} }

  \inferrule*{ b : \proc{p} \\ f : \mathsf{Func} }{ \callp{f}; b : \proc{p} }

  \inferrule*{v : \mathsf{LState}_p \to m \\ b : m \to \proc{p}}
   { \broad{v}; b : \proc{p} }

  \inferrule*{ q : \mathsf{Role} \\ b : m \to \proc{p} }
   { \switchp{q}; b : \proc{p} }

  \inferrule*{ }
    { \donep{} : \proc{p} }
\end{mathpar}
\caption{Intrinsically typed syntax of processes.}
\label{fig:processes}
\end{figure}

The type $\proc{p}$ in \cref{fig:processes} models the behavior of individual processes.
It captures the local actions a process can perform, including method execution, sending a message, receiving a message, and state-dependent branching.
The constructors of $\proc{p}$ are:

\begin{itemize}
\item $\effp{e}; b$: Call a method $e$ and continue with $b$.
\item $\send{p}{v}; b$ -- Send a message to process $p$, where the message is computed by $v : \mathsf{LState} \to m$, then continue as $b$.
\item $\recv{p}{u}; b$ -- Receive a message from process $p$, updating the local state via $u : \mathsf{LState} \times m \to \mathsf{LState}$, then continue as $b$.
\item $\broad{v}; b$ -- Broadcast a message $m$ computed by $v$, then feed $m$ to the continuation $b : m \to \proc{p}$.
\item $\switchp{p}; b$ -- Wait for a broadcast from process $p$, then continue based on the received message with $b : m \to \proc{p}$.
\item $\callp{f}; b$ -- Call procedure $f$, then continue with $b$.
\item $\donep{}$ -- Terminate.
\end{itemize}

The corresponding Lean code is:

\begin{lstlisting}
inductive Proc (p : sig.Role) : Type 1
| run (e : sig.Method p) (c : Proc p)
| send {m : Type} (q : sig.Role)
  (v : sig.LState p → m) (b : Proc p)
| recv {m : Type} (q : sig.Role)
  (u : sig.LState p → m → sig.LState p) (b : Proc p)
| shout {m : Type} (v : sig.LState p → m) (b : m → Proc p)
| listen {m : Type} (q : sig.Role) (b : m → Proc p)
| call (x : sig.Func) (b : Proc p)
| done
\end{lstlisting}

\subsection{Endpoint Projection}

The endpoint projection function $\proj{\cdot}$ maps a choreography to a network, which assigns each process its behavior extracted from the choreography.
Networks are defined as follows:
\[ \mathsf{Netw} \triangleq (p : \mathsf{Role}) \to \mathsf{Proc_p}. \]
We formally define $\proj{\cdot}$ in \cref{fig:formepp}, such
that $\proj{c}_r$ is the projection of choreography $c$ to process $r$.
When the choreography calls a method at process ${p}$, only the projection of ${p}$ will contain this method, while other processes skip it and continue projecting the continuation.
When the choreography communicates a message from ${p}$ to ${q}$, the projection constructs a send action for ${p}$, a receive action to ${q}$, and other processes proceed with the continuation.
When the choreography branches via a broadcast at process ${p}$, the projection constructs a shout for ${p}$, and a corresponding listen for other processes.
When the choreography is done, all processes are done.

\begin{figure}
\[
\begin{array}{rcl}
\proj{c} &:& \chor{} \to \mathsf{Netw} \\[6pt]
\proj{\eff{p}{e}; c'}_r &\triangleq&
  \begin{cases}
    \effp{e}; \proj{c'}_r & \text{if } r = p \\
    \proj{c'}_r           & \text{if } r \neq p
  \end{cases} \\[12pt]
\proj{\com{p}{q}{v}{u}; c'}_r &\triangleq&
  \begin{cases}
    \send{q}{v}; \proj{c'}_r & \text{if } r = p \\
    \recv{p}{u}; \proj{c'}_r & \text{if } r = q \\
    \proj{c'}_r              & \text{if } r \notin \{p,q\}
  \end{cases} \\[20pt]
\proj{\switch{p}{v}; c'}_r &\triangleq&
  \begin{cases}
    \broad{v}; \lambda m. \proj{c'(m)}_r    & \text{if } r = p \\
    \switchp{p}; \lambda m. \proj{c'(m)}_r  & \text{if } r \neq p
  \end{cases} \\
\proj{\call{f}; c}_r &\triangleq& \callp{f}; \proj{c}_r \\
\proj{\done}_r &\triangleq& \donep
\end{array}
\]
\caption{Projection from choreographies to networks.}
\label{fig:formepp}
\end{figure}

\subsection{Semantics}

In this section, we describe the operational semantics (as labelled transition systems) for choreographies and \emph{networks}, which are assignments of processes to roles.

For both choreographies and networks, a step is labeled with the type of operation that is performed and the processes that perform the step.
Concretely, the possible labels are:
\[
\mathsf{Label} ::= \effl{p} \mid \sendl{p}{q} \mid \broadl{p} \mid \calll
\]
When proving soundness, the label of a network step tells us which processes just performed an operation, allowing us to determine the corresponding choreography step.

The semantics additionally needs to evaluate procedure calls. For this, a parameter
$\pden$ is assumed, which maps each procedure to its definition.
\[\pden : \mathsf{Func} \to \mathsf{Chor}\]

\paragraph{Choreographies.}

The operational semantics for choreographies is given in \cref{fig:sem-choreographies}.
It is a binary relation over tuples of the form $\stat{c}{g}$, where $c$ is a choreography and $g : \mathsf{GState}$ a global state, which maps process identifiers to their local state:
\[
\begin{array}{lll}
\mathsf{GState} &\triangleq& (p : \mathsf{Role}) \to \mathsf{LState}_p
\end{array}
\]

A method $e$ on process $p$ is run by applying its semantics $I_p(e, g_p)$ and updating
$p$'s state accordingly.
When executing a communication $\com{p}{q}{v}{u}$, we update the state of the receiver ($q$) using $u$.

The computation rules, \textsc{Crun}, \textsc{Ccom}, \textsc{Cbcast}, and \textsc{Ccall}, correspond directly to the syntactic constructs of the choreography.
Rule \textsc{Crun} executes a method on a specific process $p$, updating the global state $g$ at $p$ via the interpretation function $I_p$.
Rule \textsc{Ccom} models directed communication; it updates the state of the receiver $q$ using the update function $u$ and the value computed from the sender $p$'s state.
Rule \textsc{Cbcast} evaluates a function on the local state of a process $p$, and then chooses how to continue the rest of the choreography based on the result.
Rule \textsc{Ccall} evaluates a call to a procedure $f$ by inserting its body into the choreography. The rule makes use of an operator $\seq$ which sequences two choreographies and which is defined in \cref{fig:seq}.
The operator $\seq$ is defined by recursion on its left argument, in a similar fashion to the standard definition of \lstinline{append} for linked lists.
Note that our semantics makes procedure calls \emph{synchronized}, i.e.
all processes are assumed to enter the procedure together. It is possible to model asynchronicity by introducing intermediate states in which only a subset of processes has entered the procedure~\cite{Montesi23}, but this is orthogonal to our use of state transformers, so
we chose to take the simpler synchronous version.

The two delay rules, \textsc{CdelayRun} and \textsc{CdelayCom}, introduce nondeterminism into the semantics.
Although the syntax of choreographies is sequential, actions involving disjoint sets of processes are independent when executed on a network.
Nondeterminism is hence needed to match the choreography semantics with the network semantics, where nondeterminism is inherent~\cite{CruzFilipeMP23}.
For example, in the program $\eff{p}{e}; \eff{q}{f}; \done{}$, if $p \neq q$, then either $p$ or $q$ could execute its method first.
To generalize this, we determine whether the processes of two operations are disjoint through a $\mathsf{disjoint}(l_1, l_2)$ predicate on labels. It is satisfied if the set of roles involved in $l_1$ does not intersect with the set of roles in $l_2$. Specifically, $\mathsf{disjoint}(\effl{p}, \effl{q})$ holds if $p \neq q$, and $\mathsf{disjoint}(\sendl{p}{q}, \effl{r})$ holds if $p \neq r \land q \neq r$, and analogously for the other combinations of methods and communications. Other combinations of labels are not considered disjoint.

\begin{figure}
\begin{mathpar}
\inferrule[Crun]{g' = g[p \mapsto I_p(e, g_p)]}
  {\stepC{\eff{p}{e}; c}{g}{\effl{p}}{c}{g'}}

\inferrule[Ccom]{g' = g[q \mapsto u(g_q, v(g_p))]}
  {\stepC{\com{p}{q}{v}{u}; c}{g}{\sendl{p}{q}}{c}{g'}}

\inferrule[CdelayRun]{\mathsf{disjoint}(\effl{p}, l) \quad \stepC{c}{g}{l}{c'}{g'}}
  {\stepC{\eff{p}{e}; c}{g}{l}{\eff{p}{e}; c'}{g'}}

\inferrule[CdelayCom]{\mathsf{disjoint}(\sendl{p}{q}, l) \quad \stepC{c}{g}{l}{c'}{g'}}
  {\stepC{\com{p}{q}{v}{u}; c}{g}{l}{\com{p}{q}{v}{u}; c'}{g'}}

\inferrule[Cbcast]{ }
  {\stepC{\switch{p}{v}; c}{g}{\broadl{p}}{c(v(g_p))}{g}}

\inferrule[Ccall]{ }
  {\stepC{\call{f}; c}{g}{\calll}{\pden(f) \seq c}{g}}
\end{mathpar}

\caption{Operational Semantics of Choreographies}
\label{fig:sem-choreographies}
\end{figure}

\begin{figure}
\[\begin{array}{rcl}
(\cdot \seq \cdot) &:& \chor{} \to \chor{} \to \chor{}
\\[.2em]
(\eff{p}{e}; c) \seq c' &\triangleq& \eff{p}{e}; (c \seq c') \\
(\com{p}{q}{v}{u}; c) \seq c' &\triangleq& \com{p}{q}{v}{u}; (c \seq c') \\
(\switch{p}{v}; c) \seq c' &\triangleq& \switch{p}{v}; (\lambda m.~(c~m) \seq c') \\
(\call{f}; c) \seq c' &\triangleq& \call{f}; (c \seq c') \\
\done{} \seq c' &\triangleq& c' \\[1em]
\end{array}\]
\caption{Sequencing of choreographies.}
\label{fig:seq}
\end{figure}

\paragraph{Networks.}

We define the operational semantics on networks in \cref{fig:sem-networks}.
It is defined on configurations of the form $\stat{n}{g}$,
where $n$ is a network and $g$ is a global state (if you prefer, a configuration is isomorphic to a map from each role to a pair of a process and a local state).
The rules for methods and communication resemble the rules for choreographies, except we now
perform steps on arbitrary processes in the network.
Note that the rules are directly defined on a network, given as a function from roles to processes.
This is in contrast to the approach taken by \citet{CruzFilipeMP23}, where networks form a partial commutative monoid, allowing them to be composed in parallel. We find our representation more suitable
for theorem proving, as we avoid the known difficulties of partial operations and furthermore do not need a congruence rule as in \citet{CruzFilipeMP23}'s approach.

The rule \textsc{Nrun} advances a single process $p$ by executing a local method and updating the network mapping $n$ so that $p$ proceeds to its continuation $b$.
The rule \textsc{Ncom} models communication: a step occurs only when a sender $p$ (ready to \procColor{send}) and a receiver $q$ (ready to \procColor{recv}) are both at the head of their respective processes. The rule updates the receiver's state and advances both processes simultaneously.
In rule \textsc{Nbcast}, process $p$ must $\procColor{shout}$ a value.
Also, \emph{every} other process $q$ in the network must be ready to \procColor{listen} for a message from $p$ (this is ensured by the function $f$, which also returns each process's continuation).
The transition updates the entire network $n$ such that $p$ and all receiving processes $q$ advance to their respective continuations based on the broadcast value $v(g_p)$.
Note that $(f (q, h)).1$ evaluates to the continuation $b'$ of a given process.
Finally, in rule \textsc{Ncall}, every process has to be in a state where it wants to call
the procedure $x$ next. Then, all processes enter the procedure. Just like in the case of choreographies, we make use of a sequencing operation. It is defined in Figure~\ref{fig:seqp},
following the same recursive approach as sequencing for choreographies.
The projected procedure environment used in \textsc{Ncall} is defined as $\mathcal{P}_p(f) := \proj{\pden(f)}_p$.

\begin{figure}
\begin{mathpar}

\inferrule[Nrun]{g' = g[p \mapsto I_p(e, g_p)] \\ n (p) = \effp {e}; {b} \\ m = n[p \mapsto b]}
  {\stepN{n}{g}{\effl{p}}{m}{g'}}

\inferrule[Ncall]{\forall p, n(p) = \callp{f}; b(p) \\
\forall p, m(p) = \mathcal{P}_p(f) \seqp b(p)}
  {\stepN{n}{g}{\calll}{m}{g}}

\inferrule[Ncom]{n (p) = \send {q} {v}; b \\ n (q) = \recv {p} {u}; c \\ m = n[p \mapsto b][q \mapsto c] \\ g' = g[q \mapsto u(g_q, v(g_p))]}
  {\stepN{n}{g}{\sendl{p}{q}}{m}{g'}}

\inferrule[Nbcast]{n (p) = \broad {v}; {b} \\ f : (q : \mathsf{Role}) \times (q \neq p) \to \{ b' \mid n (q) = \switchp {p}; {b'} \} \\ m (p) = b (v(g_p)) \\ \forall q, (h : q \neq p) \Rightarrow m (q) = (f (q, h)).1~v(g_p)}
  {\stepN{n}{g}{\broadl{p}}{m}{g}}

\end{mathpar}

\caption{Operational Semantics of Networks}
\label{fig:sem-networks}
\end{figure}

\begin{figure}
\[\begin{array}{rcl}
(\cdot \seqp \cdot) &:& \proc{p} \to \proc{p} \to \proc{p}
\\[.2em]
(\send{q}{v}; b) \seqp b' &\triangleq& \send{q}{v}; (b \seqp b') \\
(\recv{q}{u}; b) \seqp b' &\triangleq& \recv{q}{u}; (b \seqp b') \\
(\switchp{q}; b) \seqp b' &\triangleq& \switchp{q}; (\lambda m.~(b~m) \seqp b') \\
(\broad{v}; b) \seqp b' &\triangleq& \broad{v}; (\lambda m.~(b~m) \seqp b') \\
(\effp{e}; b) \seqp b' &\triangleq& \effp{e}; (b \seqp b') \\
(\callp{f}; b) \seqp b' &\triangleq& \callp{f}; (b \seqp b') \\
\donep{} \seqp b'&\triangleq& b'
\end{array}\]
\caption{Sequencing of processes.}
\label{fig:seqp}
\end{figure}

\subsection{Correctness Proofs}

Now that we have formalized the semantics of choreographies and networks as well as endpoint projection, we want to show that endpoint projection preserves semantics.
Furthermore, while arbitrary networks can contain arbitrary deadlocks, we prove that \emph{projected networks} are free from deadlocks, where a network $n$ is a projected network if there exists a choreography $c$ such that $n = \proj{c}$.

We proceed by first proving completeness (each choreography step is simulated by a network step) and soundness (each network step is simulated by a choreography step) of endpoint projection.
This ensures the semantics of choreographies and their projected networks match each other.
We then prove progress for choreographies (every choreography is either done or can do a step), which
implies progress for projected networks, by completeness.
Finally, we prove deadlock freedom as follows:
Network progress means a projected network $n$ is either done or can perform a step to a network $n'$. By soundness,
we have that $n'$ is itself a projected network, so progress applies again. By induction, we conclude that after an arbitrary number of steps, the network is either done or can perform yet another step. This is the deadlock freedom property. 

We begin with completeness. Given choreographies $c$, $c'$, the completeness theorem states that any step from $c$ to $c'$
corresponds to a step between the network projections \proj{c} and
\proj{c'}.

\begin{theorem}[Completeness]
\label{thm-completeness}
For all global states $g$, $g'$, choreographies $c$, $c'$, and labels $l$, if the choreography can step from $c$ to $c'$ with label $l$ and state change $g$ to $g'$, that is,
$\stepC{c}{g}{l}{c'}{g'}$, then the projected network can make a corresponding step
$\stepN{\proj{c}}{g}{l}{\proj{c'}}{g'}$.

\end{theorem}

\begin{proof}
By induction on the derivation of the step $\stepC{c}{g}{l}{c'}{g'}$:

\begin{itemize}
\item \textsc{Crun}: When the choreography executes a method at process $p$, this corresponds directly to a method step \textsc{Nrun} in the network projection at the same process.
\item \textsc{Ccom}: Sending a message from process $p$ to $q$ in the choreography corresponds to a sending and receiving step \textsc{Ncom} between the same processes in the network projection. 
\item \textsc{CdelayRun} and \textsc{CdelayCom}:
The derivation of a delay step recursively contains another step. By the induction hypothesis, such a step can be mapped to a network step. The disjointness conditions ensure that the delayed actions do not interfere with the simulated step on the network projection.
\item \textsc{Cbcast}: Broadcast steps in the choreography correspond to broadcast steps in the network \textsc{Nbcast} after projection.
\item \textsc{Ccall}: Call steps in the choreography correspond to
call steps in the network \textsc{Ncall} after projection.
\qedhere
\end{itemize}
\end{proof}

In the case for \textsc{Ccall},
applying the \textsc{Ncall} step involves proving that its second precondition is fulfilled,
which in this case amounts to showing that, for any role $p$,
\[ \proj{\pden(x) \seq c}_p = \proj{\pden(x)}_p \seqp \proj{c}_p. \]
This is an immediate consequence of the following lemma.
\begin{lemma}[Projection distributes over composition]
For choreographies $c, c'$ and role $p$, we have $\proj{c \seq c'}_p = \proj{c}_p \seqp \proj{c'}_p.$
\end{lemma}
\begin{proof}
By induction on $c$.
\end{proof}

Having proven completeness, we now show the converse:
The soundness theorem states that any step taken by a projected network corresponds to a valid step in the choreography semantics.
While completeness maps choreography steps to network steps, soundness cannot map
arbitrary network steps to choreography steps, since the network semantics allows for more general behavior. Hence, we require that the network is the result of projecting a choreography. In the proof, for any network step, we need to inductively traverse this choreography to find the corresponding choreography step.

\begin{theorem}[Soundness]
For any global states $g$, $g'$, choreography $c$, network $n$, and label $l$, if the network starting from the projection of $c$ can step to $n$ with label $l$ and state change $g$ to $g'$, i.e.,
$\stepN{\proj{c}}{g}{l}{n}{g'}$, then there exists a choreography $c'$ such that $n = \proj{c'}$ and the choreography can step correspondingly: $\stepC{c}{g}{l}{c'}{g'}$.

\end{theorem}

\begin{proof}
By induction on the structure of the choreography $c$:

\begin{itemize}
\item $\eff{p}{e}; c'$ or $\com{p}{q}{v}{u}; c'$: We check if the network step matches the operation. If so, we know that the first operation of the choreography is executed, and return the corresponding choreography step (\textsc{Crun} or \textsc{Ccom}, respectively). Otherwise, we apply one of the delay rules to skip the current operation. We then inductively search the rest of the choreography ($c'$) for an operation that matches the network step. 
\item $\switch{p}{v}; c'$: If the first operation is a broadcast, then all projected processes must either have \procColor{shout} or \procColor{listen} as their first operation. Therefore, the network step can only be a broadcast step \textsc{Nbcast}. We map it to the corresponding choreography step \textsc{Cbcast}.
\item $\call{f}; c'$: Similarly, if the first operation is a procedure call, then all projected processes must have the message call as their first operation. Therefore, the network step can only be \textsc{Ncall}, which we map to \textsc{Ccall}.
\item $\done{}$: No step can be done, so the property holds vacuously.
\qedhere
\end{itemize}

\end{proof}

We have shown the correctness of endpoint projection with respect to the operational semantics.
We are now left to show progress and deadlock freedom.

The progress lemma for choreographies asserts that any choreography is either finished ($\done{}$) or can perform a step in the operational semantics.

\begin{lemma}[Choreography Progress]
\label{thm-chor-progress}
For any choreography $c$ and global state $g$:
\[(c = \done{}) \lor (\exists c', l, g'.~\stepC{c}{g}{l}{c'}{g'})\]
In words: Either the choreography $c$ is already terminated, or there exists a label $l$, a next choreography $c'$, and an updated global state $g'$, such that $\stat{c}{g}$ can take an $l$-labeled step to $\stat{c'}{g'}$.
\end{lemma}

\begin{proof}
By case distinction on the structure of the choreography:
If the choreography has form $\com{p}{q}{v}{u};c $, $\switch{p}{v}; c$, $\eff{p}{e}; c$, or $\call{f}; c$,
there is an obvious corresponding step in the operational semantics that executes the first
operation in the choreography.
If the choreography is $\done{}$, the theorem holds trivially.
\end{proof}

The progress lemma for the network semantics guarantees that the projection of any choreography either has all processes terminated or can make a network step according to the operational semantics.

\begin{lemma}[Network Progress]
\label{thm-net-progress}
For any global state $g$ and choreography $c$:
\[(\forall p.~\proj{c}_p = \donep{}) \lor (\exists c', l, g'.~\stepN{\proj{c}}{g}{l}{\proj{c'}}{g'})\]

\noindent
That is, either every process in the projected network is $\donep{}$, or there exists a network step labeled $l$ from $\proj{c}$ at state $g$ to the projected next choreography at the updated global state $g'$.
\end{lemma}

\begin{proof}
The result follows from progress (\cref{thm-chor-progress}) and completeness (\cref{thm-completeness}) for choreographies:
\begin{itemize}
\item By the \emph{progress} lemma (\cref{thm-chor-progress}), the choreography $c$ is either done or can make a step from $c$ to $c'$ labeled $l$ from state $g$ to $g'$.
\item If $c$ is done, then by definition all processes in the projected network are also done.
\item If $c$ can make a step, then by \emph{completeness} (\cref{thm-completeness}), the projected network \proj{c} can make a corresponding network step.
\item Thus, the network semantics either terminates or can make progress in lockstep with the choreography semantics.
\qedhere
\end{itemize}
\end{proof}

The deadlock freedom theorem states that any reachable state of the network semantics, derived from a choreography, is either terminated or can make a further step; hence, no deadlock occurs.
For this theorem, we use the relation $\stepNS{n}{g}{n'}{g'}$, which is the reflexive-transitive closure of the network step relation.

\begin{theorem}[Deadlock Freedom]
For any global states $g$, $g'$, choreography $c$, and network state $n$ such that the projected choreography \proj{c} steps star-wise (zero or more steps) to $n$ at state $g'$:
\[\begin{array}{l}
(\stepNS{\proj{c}}{g}{n}{g'}) \Rightarrow \\
\quad(\forall x.~n~x = \donep{}) \lor ( \exists c', l, g'.~\stepN{n}{g}{l}{\proj{c'}}{g'})
\end{array}\]

\noindent
In words: From any reachable network state $n$ after zero or more steps, either all processes are done or the network can make a further step; thus, the network semantics is deadlock-free.
\end{theorem}

\begin{proof}
We strengthen the conclusion to the following invariant, which we prove to hold by induction over the series of network steps $(\stepNS{\proj{c}}{g}{n}{g'})$:
\[\begin{array}{l}
(\exists c.~(\proj{c} = n)) \land \\
\quad( (\forall x.~n~x = \donep{}) \lor \exists c', l, g'.~\stepN{n}{g}{l}{\proj{c'}}{g'}).
\end{array}\]
For this, it is necessary to show that (i) a projected choreography satisfies the invariant, and (ii) the invariant is preserved by a network step.
Part (i) follows from network progress, while (ii) follows from network progress and soundness.
\end{proof}

\subsection{Interpreter}
\label{sec:interp}

In this subsection, we prove partial correctness of a deterministic interpreter for choreographies. More precisely, if the interpreter successfully
returns a result, then this result is reachable in the operational semantics.
The interpreter we verify is defined as follows:

\begin{lstlisting}
def Chor.interp (d : sig.Func → Chor) (g : GState) : Chor → Nat → Option GState
| run p e c, n+1 =>
  c.interp d g[p ↦ sig.I p e (g p)] n
| com p q _h v u c, n+1 =>
  c.interp d g[q ↦ u (g q) (v (g p))] n
| bcast p v c, n+1 => (c (v (g p))).interp d g n
| call x c, n+1 => (d x).interp d g n >>= (c.interp d . n)
| done, _+1 => g
| _, 0 => .none
\end{lstlisting}

Because our choreographic language allows nonterminating programs, our interpreter is necessarily non-total.
We use a standard \emph{fuel}-based approach~\cite{AppelM01, Owens16, Siek13, Amin17} to encode this
fact in Lean\footnote{We only use fuel for the interpreter: The operational semantics does not need it, as it is are based on a relational, small-step approach, and the Hoare logic does not need it as it only proves \emph{partial} correctness}.
The interpreter's partiality is indicated by its \lstinline{Option} return type. It takes a fuel argument of type \lstinline{Nat} setting a limit on the recursion depth.
The fuel argument ensures that the function is structurally recursive, despite the nested recursion in the \lstinline{call} case, which would not be well-founded otherwise.
While we previously defined a simpler interpreter in Section~\ref{s:int}, the version
shown here is more suitable for formal reasoning, as it supports proofs by natural induction.

To begin, we need some lemmas about sequencing.
We first prove that sequencing is associative. We will need this as the semantics
for procedure calls introduces the sequencing operator into our programs.

\begin{lemma}[Associativity of Sequencing]
\label{lem:assoc}
For choreographies $c, d, e$, we have
\[ (c \seq d) \seq e = c \seq (d \seq e). \]
\end{lemma}
\begin{proof}
By induction on $c$.
\end{proof}

Now, we prove two lemmas showing that sequencing is compatible with the operational semantics.
The first lemma states that if a reduction step is valid on a choreography, it remains valid
when code is added to the end of the choreography (using the sequencing operator $\seq$).

\begin{lemma}
\label{lemma:lift1}
For choreographies $c, c', c''$, global states $g, g'$, and label $l$, if~$\stepC{c}{g}{l}{c'}{g'}$, then
$\stepC{c \seq c''}{g}{l}{c' \seq c''}{g'}$.
\end{lemma}
\begin{proof}
By induction on the operational semantics rules.
Each rule is essentially mapped to itself, so we just discuss the \textsc{Crun} case, as the others
are similar. By assumption, we have $\stepC{\eff{p}{e}; c'}{g}{l}{c'}{g'}$ for some
$c', l, g, g'$, and want to show $\stepC{(\eff{p}{e}; c') \seq c''}{g}{l}{c' \seq c''}{g'}$.
By definition of sequencing, this is equal to
$\stepC{\eff{p}{e}; (c' \seq c'')}{g}{l}{c' \seq c''}{g'}$, which is a direct consequence of the
\textsc{Crun} rule.
In the case of the \textsc{Ccall} rule, we additionally have to use associativity of sequencing.
\end{proof}

The second lemma state that this property also holds for the reflexive transitive closure of the choreography semantics, written $\stepCS{c}{g}{c'}{g'}$.

\begin{lemma}
\label{lemma:lift}
For choreographies $c, c', c''$, global states $g, g'$, and label $l$, if
$\stepCS{c}{g}{c'}{g'}$, then
$\stepCS{c \seq c''}{g}{c' \seq c''}{g'}$.
\end{lemma}
\begin{proof}
By induction on the derivation of $\stepCS{c}{g}{c'}{g'}$. The base case is trivial,
the step case uses Lemma~\ref{lemma:lift1}.
\end{proof}

We can finally prove partial correctness of our interpreter.

\begin{theorem}[Partial Correctness of Interpreter]
For choreography $c$, global states $g, g'$, and natural number $n$, if
the interpreter given $n$ fuel evaluates $c$ in state $g$ to a state $g'$, then
$\stepCS{c}{g}{\done}{g'}$.
\end{theorem}
\begin{proof}
We want to return $\stepCS{c}{g}{\done}{g'}$.
First, we perform induction on $n$. If $n = 0$, the interpreter cannot return a state, so the
statement hold vacuously. If $n = m+1$, we perform a case distinction on $c$.

(i) $\eff{p}{e}; c'$ or $\com{p}{q}{v}{u}; c'$ or $\switch{p}{v}; c'$:
The induction hypothesis gives us \[\stepCS{c''}{g''}{\done}{g'},\] where
$\stat{c''}{g''}$ can be reached from $\stat{c'}{g'}$ in one step, by using the
rule in the semantics corresponding to the construct (\textsc{Crun}, \textsc{Ccom}, or \textsc{Cbcast}).
We compose the respective rule (schematically: $\stepC{c}{g}{l}{c''}{g''}$) with the induction hypothesis to get $\stepCS{c}{g}{\done}{g'}$.

(ii) $\done{}$: By reflexivity, we immediately have\\
$\stepCS{\done}{g}{\done}{g}$.

(iii) $\call{f}; c'$: In the interpreter, this case involves nested recursion,
so we first perform a case distinction on the inner recursive call. If it fails, then the whole call fails, so the statement hold vacuously.
Otherwise, the idea is that reducing $\call{f}; c'$ consists of three steps:
First, \textsc{Ccall} is applied. Second, $\pden{f} \seq c'$ is reduced. Third,
$c'$ is reduced.

We get the reduction of $\pden{f}$ via the induction hypothesis,
which we then lift with Lemma~\ref{lemma:lift} to a reduction
on $\pden{f} \seq c'$. We compose \textsc{Ccall}, the reduction on $\pden{f} \seq c'$, and the reduction on $c'$ (which we also get via the induction hypothesis), to get the complete execution.

\end{proof}

\begin{theorem}[Interpreter is Monotonous]
If the interpreter evaluates to state $g'$ with $n$ fuel, it also evaluates to $g'$ with $n+1$ fuel.
\end{theorem}
\begin{proof}
By induction over the definition of the interpreter.
\end{proof}

\subsection{Confluence}
\label{sec:confluence}

As discussed, our semantics is non-deterministic. We show that the non-determinism is well-behaved by proving confluence, which is also relevant the validity of our Hoare logic. For the mechanization in this subsection and in Section~\ref{sec:hoare}, we made use of LLMs (Github Copilot with GPT 5.4).

We first show that computation steps with different labels are independent, meaning they do not block each other, and it does not matter in which order we execute them.

\begin{lemma}[Diamond]
If $\stepC{c}{g}{l_1}{c_1}{g_1}$ and\\$\stepC{c}{g}{l_2}{c_2}{g_2}$ with $l_1 \neq l_2$, there exist $c_3, g_3$ such that $\stepC{c_1}{g_1}{l_2}{c_3}{g_3}$
and\\$\stepC{c_2}{g_2}{l_1}{c_3}{g_3}$.
\end{lemma}

Furthermore, all the non-determinism is contained in the label:

\begin{lemma}
If $\stepC{c}{g}{l}{c'}{g'}$ and $\stepC{c}{g}{l}{c''}{g''}$, then
$c' = c''$ and $g' = g''$.
\end{lemma}

Also, if two steps are possible, they either have the same or disjoint labels (in the first case, they must be the same step, by the lemma above):

\begin{lemma}
If $\stepC{c}{g}{l_1}{c'}{g'}$ and $\stepC{c}{g}{l_2}{c''}{g''}$
then either $l_1 = l_2$ or $\mathsf{disjoint}(l_1, l_2)$.
\end{lemma}

We can now show confluence, i.e. that any two reduction sequences can be joined.

\begin{theorem}[Confluence]
If $\stepCS{c}{g}{c_1}{g_1}$ and\\$\stepCS{c}{g}{c_2}{g_2}$, then
there exist $c_3, g_3$ such that\\$\stepCS{c_1}{g_1}{c_3}{g_3}$ and
$\stepCS{c_2}{g_2}{c_3}{g_3}$.
\end{theorem}
\begin{proof}
We show semi-confluence by induction on $\chorColor{\chorColor{\to}^\star}$ and using the three lemmas above, from which confluence follows by a standard argument
\end{proof}

We derive the following lemma showing that if a configuration terminates, then reducing it once leads to a configuration that still terminates with the same result.
This lemma will be useful when verifying several of the Hoare logic rules.

\begin{lemma}
\label{lem:stepdone}
If $\stepC{c}{g}{l}{c'}{g'}$ and $\stepCS{c}{g}{\done{}}{g''}$, then
$\stepCS{c'}{g'}{\done{}}{g''}$.
\end{lemma}
\begin{proof}
Follows from confluence and the fact that no steps are possible in state $\stat{\done{}}{g''}$.
\end{proof}

Intuitively, the reason why this lemma is useful is that Hoare logic makes assertions about \emph{all} (terminating) executions. With this lemma, we need only consider one execution and transport its results to the other ones.

\subsection{Hoare Logic}
\label{sec:hoare}

In this subsection, we define a partial correctness Hoare logic for choreographies.
We first define a judgment $\vdash_\pden \{P\}~c~\{Q\}$ semantically: The judgment is valid
when every terminating execution of $c$ (using a procedure interpretation $\pden$) that starts in a state satisfying precondition $P$ ends in a state satisfying postcondition $Q$. Both $P$ and $Q$ are predicates over global states.

An example of using Hoare logic is shown in Figure~\ref{fig:hoare-ex}. We assume processes $p, q, r$ each of whose state is a natural number. In the program $p, q$ send their respective states to $r$, which adds them to its own state. We prove that if $r$'s state is equal to 0 at the beginning, it is equal to the sum of $p$'s and $q$'s state at the end (the example uses no procedures, so we omit the procedure interpretation $\pden$).

\begin{figure*}[t]
\begin{mathpar}
\inferrule*[right=Hoare-Com]{\inferrule*[right=Hoare-Com]{\inferrule*[right=Hoare-Done]{ }{\vdash \{\lambda g.~g_r = g_p + g_q\}~\done{}~\{\lambda g.~g_r = g_p + g_q\}}}{\vdash \{\lambda g. g_r = g_p\}~\com{q}{r}{\lambda n.~n}{\lambda s~n.~s+n};  \done{}~\{\lambda g.~g_r = g_p + g_q\}}}{\vdash \{\lambda g.~g_r = 0\}~\com{p}{r}{\lambda n.~n}{\lambda s~n. s+n};\com{q}{r}{\lambda n.~n}{\lambda s~n.~s+n};\done{}~\{\lambda g.~g_r = g_p + g_q\}}
\end{mathpar}
\caption{Verification, in Hoare logic, of a simple addition protocol. Technically, both applications of the communication rule are preceded by an application of the consequence rule;  we omit these for readability.}
\label{fig:hoare-ex}
\end{figure*}

We now give a formal, semantic definition of Hoare logic judgments.

\begin{definition}[Hoare Triple]
\label{def:triple}
\[
\vdash_\pden \{P\}~c~\{Q\} \triangleq \forall g, g'.~ P(g) \Rightarrow \stepCS{c}{g}{\done{}}{g'} \Rightarrow Q(g')
\]
\end{definition}

We then state the Hoare logic rules as theorems about this judgment. The rules are listed in Figure~\ref{fig:hoare}.
\textsc{Hoare-Conseq} allows strengthening preconditions and weakening postconditions.
\textsc{Hoare-Done} says that $\done{}$ does not change the state.
\textsc{Hoare-Seq} al lows composing Hoare triples using $\seq$.
\textsc{Hoare-Run} captures the effect of $\eff{p}{e}$ on $p$'s state,
and \textsc{Hoare-Com} does the same for the effect of $\com{p}{q}{v}{u}$ on $q$'s state.
In \textsc{Hoare-Bcast}, we need to prove the postcondition holds for possible
continuations $c(m)$ (given a message $m$).
\textsc{Hoare-Call} expresses that $\call{f}; c$ is essentially equivalent to
$\pden{f} \seq c$ and we hence need to proof triples for both $\pden{f}$ and $c$.

\begin{figure}
\begin{mathpar}
\inferrule[Hoare-Conseq]{\vdash_\pden \{P'\}~c~\{Q'\} \\ \forall g.~P(g) \Rightarrow P'(g) \\ \forall g.~Q'(g) \Rightarrow Q(g)}
  {\vdash_\pden \{P\}~c~\{Q\}}

\inferrule[Hoare-Done]{ }
  {\vdash_\pden \{P\}~\done{}~\{P\}}

\inferrule[Hoare-Seq]{\vdash_\pden \{P\}~c~\{Q\} \\ \vdash_\pden \{Q\}~c'~\{R\}}
  {\vdash_\pden \{P\}~c \seq c'~\{R\}}

\inferrule[Hoare-Run]{\vdash_\pden \{P\}~c~\{Q\}}
  {\vdash_\pden \{\lambda g.~P(g[p \mapsto I_p(e, g_p)])\}~\eff{p}{e}; c~\{Q\}}

\inferrule[Hoare-Com]{\vdash_\pden \{M\}~c~\{Q\}}
  {\vdash_\pden \{\lambda g.~M(g[q \mapsto u(g_q, v(g_p))])\}~\com{p}{q}{v}{u}; c~\{Q\}}

\inferrule[Hoare-Bcast]{\forall m.~\vdash_\pden \{M(m)\}~c(m)~\{Q\}}
  {\vdash_\pden \{\lambda g.~M(v(g_p), g)\}~\switch{p}{v}; c~\{Q\}}

\inferrule[Hoare-Call]{\vdash_\pden \{P\}~\pden(f)~\{M\} \\ \vdash_\pden \{M\}~c~\{Q\}}
  {\vdash_\pden \{P\}~\call{f}; c~\{Q\}}
\end{mathpar}
\caption{Hoare logic rules for choreographies.}
\label{fig:hoare}
\end{figure}

To prove the soundness of the \textsc{Hoare-Seq} rule, we need to be able to split a terminating execution on $c \seq{} c'$ into two execution on $c$ and $c'$, so that we can use the premises of the rule.

\begin{lemma}[Split $\labelC{\star}$]
\label{lem:split}
If $\stepCS{c \seq{} c'}{g}{\done{}}{g''}$ then there exists a state
$g'$ such that $\stepCS{c}{g}{\done{}}{g'}$ and $\stepCS{c'}{g'}{\done{}}{g''}$.
\end{lemma}
\begin{proof}
By strong induction on the length of the execution $\stepCS{c \seq{} c'}{g}{\done{}}{g''}$.
\end{proof}

The proof uses the notion of the length of an execution, which we formalize using the following relation saying that a choreography reduces to another one in $n$ steps:

\begin{definition}[N-Step Relation]
The n-step relation\\$\stepM{c}{g}{n}{c'}{g'}$
is inductively defined by

\begin{mathpar}
\stepM{c}{g}{0}{c}{g} \\
\inferrule{\stepC{c}{g}{l}{c'}{g'} \\ \stepM{c'}{g'}{n}{c''}{g''}}
{\quad\stepM{c}{g}{n+1}{c''}{g''}}
\end{mathpar}
\end{definition}

To use this relation, we need to map from $\chorColor{\twoheadrightarrow}^n$ to $\chorColor{\to}^l$ and back, which we can do by induction on the respective derivation.

With lemma~\ref{lem:split} in place, we can verify our Hoare logic.

\begin{theorem}
The rules in Fig.~\ref{fig:hoare} are valid theorems about the judgment in Definition~\ref{def:triple}.
\end{theorem}
\begin{proof}
We go through each theorem:
\begin{itemize}
\item \textsc{Hoare-Conseq}: Follows directly from Definition~\ref{def:triple}.
\item \textsc{Hoare-Done}: From $\stepCS{\done{}}{g}{\done{}}{g'}$ follows $g = g'$, and hence $P(g)$ implies $P(g')$.
\item \textsc{Hoare-Seq}: Apply Lemma~\ref{lem:split} and the assumptions.
\item \textsc{Hoare-Run, Hoare-Com, Hoare-Bcast, Hoare-Call}: In these cases, we know
the program can make a step to $\stat{c}{g}$ and (by assumption) that it evaluates to $\stat{\done{}}{g'}$ (for some $c, g, g'$).
Applying Lemma~\ref{lem:stepdone} we get that $\stat{c}{g}$ also evaluates to $\stat{\done{}}{g'}$
\end{itemize}
\end{proof}

\section{Related Work}
\label{sec:rel}

\paragraph{Choreographies}

\citet{Montesi13} introduced choreographic programming; Montesi has also written
an introductory text book~\cite{Montesi23}.

\citet{CruzFilipeMP23} (improving on previous work~\cite{Cruz-FilipeMP21}) mechanize a choreographic language, which compared to ours, has label selection (and hence deals with \emph{knowledge of choice}), but no role-dependent methods or Hoare logic. They represent states of processes as assignments of values to variables; our approach can model this approach (by defining states to be functions from a type representing variables to a type representing values). However, our approach is simpler in not requiring variables or local expressions in the core calculus.

\citet{HirschG22} define a calculus for \emph{higher-order} typed choreographies, parametrized over a local language. They use de Bruijn syntax (and hence notions like substitution and weakening), while our local expressions are represented as metalanguage (Lean) functions.

While choreographies are usually defined as \emph{languages}, \citet{ShenKK23} implement choreographic programming as a Haskell \emph{library} \emph{HasChor}. They feature a broadcasting construct for making a value available throughout the choreography, which our broadcasting is based on.

\citet{Cruz-FilipeM16b} describe how to implement a number of concurrent computation-heavy algorithms. Their communication construct is defined as a stateful update very similar to ours,
but the paper focusses on implementing algorithms and does not provide a formal semantics.

\citet{Pohjola22} mechanically verify a compiler from a choreographic language to an ML dialect. Their mechanization also allows computations to be performed by arbitrary metalanguage functions, but still uses an explicit context with variables and binders. They focus on semantic preservation for their compiler, while we treat broadcasting, role-dependent methods, confluence, and Hoare logic. 

\citet{Cruz-FilipeGMP23} describe a Hoare logic for choreographies; a mechanization was provided in work concurrent to ours~\cite{cruzfilipe26}. Their Hoare logic is a deep embedding, defining state formulas, substitution, signatures, and rules syntactically. We instead define Hoare logic semantically, representing rules as 
Lean theorems, which simplifies the mechanization and allows existing Lean tactics to be used more directly.

\paragraph{Session Types}

\emph{Binary} session types describe messages between two communicating participants~\cite{Honda93},
ensures that well-typed programs do not deadlock.
A later extension are \emph{multiparty} session types~\cite{Bettini08,CastagnaDP11, Honda16}, which are similar to choreographies in giving a unified description of a distributed
protocol. However, they are, as the name implies, types, and hence separate the specification of a protocol from its implementation.

Thiemann mechanized deadlock freedom for binary session types using a de Bruijn representation~\cite{Thiemann19}, and
later introduced the state transformer approach to embedding session types in a functional language~\cite{Thiemann23}, focussing on binary session types, and arguing for deadlock freedom without mechanized proofs.

\citet{Castro-PerezFJ26} mechanize a novel approach to multiparty session types where processes are directly checked against the global type, without projection.

\citet{TiroreBC25} point out flaws in the proof in the original multiparty session types paper~\cite{HondaYC08} and provide a mechanization of soundness for a restricted theory.
\citet{EkiciKY25} mechanize multiparty session types with subtyping, pointing out unsoundness in several
previous papers on the same topic.

\paragraph{HOAS} The advantages of modeling binders using metalanguage functions feature in works on HOAS~\cite{PfenningE88} and parametric HOAS~\cite{Chlipala08}.
Sadly, standard HOAS, where bindings are represented using functions \lstinline|Term → Term|, is not usable in theorem provers like Lean as it violates strict positivity, which is why parametric HOAS was developed. Parametric HOAS represents bindings using functions \lstinline|v → Term|, where \lstinline|v| is parametric. In contrast to both approaches, our metalanguage functions take \emph{values} as arguments, which is known as a mixed embedding~\cite{Chlipala21}.

\paragraph{Effects} Our definition of broadcasting uses a continuation $m \to \chor$ that depends on the broadcasted message. A similar approach is used to encode effectful operations in mixed embeddings for effectful computations, like free(r) monads~\cite{Kiselyov15}
and interaction trees~\cite{Xia20, Seassau25}.
The idea of free monads is that an effectful computation is represented as a tree, where every non-leaf node is an effectful operation, with one child for each possible return value of the operation. For example, our $\chorColor{bcast}$ takes a continuation \lstinline|m \to \chor|, which corresponds to an effectful operation that returns $m$. The leaf nodes represent the return value of the overall computation, loosely corresponding to our $\done{}$ construct.

It would be interesting to make this connection more concrete in future work:
Does the $\seq$ function correspond to monadic bind, and Lemma~\ref{lem:assoc} to the standard monadic associativity law?
Further, if one views the process language as another free monad, then one might consider endpoint projection as an unusual kind of effect handler~\cite{Shen2024, Bauer18} that returns a \emph{set} of computations,
making concurrency explicit.

\section{Conclusion}
\label{sec:conc}

We have demonstrated the utility of the state transformer approach for mechanization of distributed programming by formalizing a choreographic language with point-to-point communication, broadcasting, recursive procedures, and role-specific local methods.
The state transformer approach avoids dealing with variable binding issues and decouples the distributed and the local parts of computation by letting the meta-language handle the latter.
We proved soundness and completeness of endpoint projection, as well as deadlock freedom for projected networks.

\paragraph{Limitations and Future Work}

Our formalization of broadcasting involves the entire network. Similarly, recursive procedure calls involve synchronization. While some practical approaches like HasChor~\cite{ShenKK23} use this form of broadcasting, there are others which use the more fine-grained \emph{knowledge of choice}~\cite{Castagna12} approach, and extending our mechanization
to support this concept or the more advanced multiply-located values and enclaves~\cite{BatesN24} would improve its practical relevance.
A convenience issue of our formalization is that local states have to be defined explicitly; it would be better if they were inferred automatically from the choreography.

\section*{Acknowledgments}

This research work was supported by the National Research Center for Applied Cybersecurity ATHENE. ATHENE is funded jointly by the German Federal Ministry of Education and Research and the Hessian Ministry of Higher Education, Research and the Arts.
This work was funded by the LOEWE initiative (Hesse, Germany)\\ \ [LOEWE/4a//519/05/00.002(0013)/95],
by the LOEWE initiative (Hesse, Germany) within the
emergenCITY center [LOEWE/1/12/519/03/05.001(0016)/72], and by the Deutsche Forschungsgemeinschaft (DFG, German Research Foundation) as part of Germany's Excellence Strategy –- EXC-3057/1 ``Reasonable Artificial Intelligence'' –- Project No. 533677015.

\bibliographystyle{ACM-Reference-Format}
\bibliography{bib}

\end{document}